\documentclass[12pt,reqno]{amsart}
\usepackage{amsmath}
\usepackage{latexsym}
\usepackage{amsfonts}
\usepackage{amssymb}
\usepackage{color}
\usepackage{bbm,dsfont}
\usepackage{graphicx}
\usepackage{hyperref}
\usepackage{enumerate}

\newtheorem{proposition}{Proposition}

\theoremstyle{definition}
\newtheorem{example}{Example}

\newcommand{\C}{\mathbb C} 
\newcommand{\mo}[1]{\left| #1 \right|} 

\newcommand{\hi}{\mathcal{H}} 
\newcommand{\hik}{\mathcal{K}} 
\newcommand{\K}{\mathcal{K}} 
\newcommand{\lh}{\mathcal{L(H)}} 
\newcommand{\lk}{\mathcal{L(K)}} 
\newcommand{\trh}{\mathcal{T(H)}} 
\newcommand{\trhk}{\mathcal{T}(\mathcal{H}\otimes \mathcal{K})} 
\newcommand{\sh}{\mathcal{S(H)}} 
\newcommand{\ph}{\mathcal{P(H)}} 
\newcommand{\sk}{\mathcal{S(K)}} 

\newcommand{\ip}[2]{\left\langle\,#1\,|\,#2\,\right\rangle} 
\newcommand{\tr}[1]{\textrm{tr}\left[#1\right]} 
\newcommand{\ptr}[1]{\textrm{tr}_{\mathcal{K}}[#1]} 
\newcommand{\hptr}[1]{\textrm{tr}_{\mathcal{H}}[#1]} 
\newcommand{\id}{\mathbbm{1}} 

\newcommand{\mc}[1]{\mathcal{#1}} 

\newcommand{\Ao}{\mathsf{A}}
\newcommand{\Eo}{\mathsf{E}}
\newcommand{\Fo}{\mathsf{F}}
\newcommand{\So}{\mathsf{S}}
\newcommand{\Zo}{\mathsf{Z}}

\newcommand{\Ec}{\mathcal{E}}
\newcommand{\Fc}{\mathcal{F}}
\newcommand{\Uc}{\mathcal{U}} 
\newcommand{\Vc}{\mathcal{V}} 

\newcommand{\R}{\mathbb{R}}

\begin{document}

\title[Notes on Deterministic Programming of Quantum...]{Notes on Deterministic Programming of Quantum Observables and Channels}

\author[Heinosaari]{Teiko Heinosaari}
\author[Tukiainen]{Mikko Tukiainen}
\address{\textbf{Teiko Heinosaari}; Turku Centre for Quantum Physics, Department of Physics and Astronomy, University of Turku, Finland}
\email{teiko.heinosaari@utu.fi}
\address{\textbf{Mikko Tukiainen}; Turku Centre for Quantum Physics, Department of Physics and Astronomy, University of Turku, Finland}
\email{mikko.tukiainen@utu.fi}

\begin{abstract}
We study the limitations of deterministic programmability of quantum circuits, e.g.,\ quantum computer. More precisely, we analyse the programming of quantum observables and channels via quantum multimeters. We show that the programming vectors for any two different sharp observables are necessarily orthogonal, whenever post-processing is not allowed. This result then directly implies that also any two different unitary channels require orthogonal programming vectors. This approach generalizes the well-known orthogonality result first proven by Nielsen and Chuang. In addition, we give size-bounds for a multimeter to be efficient in quantum programming.
\end{abstract}

\maketitle

\section{Introduction}\label{sec:intro}

A common computer consists of a fixed number of gates that process classical information by performing classical operations on input data, strings of bits. 
A single computer can be used for many different tasks simply by changing its program. 
Indeed, common computer can be programmed to perform arbitrary classical manipulations on the input data, although the execution time depends on its computational speed.
This \emph{universality}, making it a very versatile device, is arguably the most important feature of a computer.

A quantum computer presumably consists of \emph{quantum gates} that are used to implement desired quantum operations on a quantum input, for instance \emph{qubits}. It's suspected, that large scale quantum computers should surpass their classical counterparts with superior speed and hence revolutionize the world by opening for study a whole class of computational tasks that ordinary computers simply cannot grasp efficiently, such as factorization of large integers and simulation of quantum systems.
However, there is, besides the practical issues of actually building one, a major theoretical problem with quantum computer, namely its universality. 

Contrary to classical computing, where it is possible to design a universal gate array that can be \emph{deterministically programmed} to implement any arbitrary function on a given classical input data, in quantum case no such universality holds. 
Indeed, it was proven by Nielsen and Chuang in \cite{NiCh97}, that to program a quantum gate array to realize any two different \emph{unitary channels}, being the mathematical representatives of quantum gates, requires orthogonal programming vectors. A cardinality argument leads to impossibility of implementing all unitary channels even in an infinite dimensional separable Hilbert space. The existence of imperfect, that is probabilistic, universally programmable quantum gate arrays has, however, been confirmed and studied in \cite{NiCh97,HiBuZi01,HiZiBu06}. 

Another useful device would be a \emph{quantum multimeter} that could be programmed to measure any quantum observable. 
Also this universally programmable device is feasible only if some error is accepted or the implementation is probabilistic \cite{DuBu02,FiDuFi02,DaPe05,PerezGarcia06}. 
It turns out  that the reason for this is very similar to the impossibility result on unitary channels: the programming of a multimeter to realize any two different \emph{sharp observables} demands orthogonal programming vectors. 
This orthogonality result can however be alleviated -- so that different sharp observables can be programmed with non-orthogonal programming vectors -- assuming that classical \emph{post-processing} of the measurement statistics is allowed \cite{ZiBu05}.

As we will show, all the previously mentioned no-go results have natural formulations in the context of quantum theory of measurement. 
This wider framework also reveals their inherent connections and opens up some immediate generalizations. 

Our study is organised as follows.
In section \ref{sec:prelim} we start by introducing some preliminary concepts and mathematical machinery. 
In section \ref{sec:prog} we study the limitations of deterministic programmability of quantum channels and observables. 
We start by first studying the programmability of sharp observables and show that the programming vectors for any two different sharp observables are necessarily orthogonal, whenever post-processing is not allowed. 
We show that this result directly implies that also two different unitary channels require orthogonal programming vectors, giving an alternative proof for the well-known result first proven in \cite{NiCh97}, yet in a more general context.
We also note further limitations: the programming of a sharp observable and an \emph{extreme observable} require orthogonal programming vectors. The same result holds also for a unitary and an extreme channel. In fact, the programming limitations of sharp observables and unitary channels are somewhat analogous: a difference arises only when post-processing of the measurement statistics is allowed.
In the final section \ref{sec:eff} we develop the concept of \emph{effectiveness} for deterministic programming by giving bounds for a programming protocol to be efficient. 

\section{Preliminaries}\label{sec:prelim}

In this section we fix notations and recall the basic concepts needed in our investigation.
Let $\hi$ be a complex separable Hilbert space with either finite or countably infinite dimension. 
We denote by $\lh$ the set of \emph{bounded linear operators}, by $\trh$ the set of \emph{trace class operators} and by $\ph$ the set of \emph{projections} on $\hi$. We identify the set of \emph{quantum states} as $\sh = \{\rho \in \trh \, | \, \rho \geq 0, \tr{\rho}=1\}$. In particular a state $\rho$ is called \emph{pure}, or equivalently a \emph{vector state}, if $\rho \in \sh \cap \ph$, that is $\rho = P[\varphi] = |\varphi\rangle\langle\varphi|$ for some unit vector $\varphi \in \hi$.

\subsection{Observables}

Let $\Omega$ be a nonempty set and $\Sigma \subset 2^\Omega$ a $\sigma$-algebra. 
We denote by $\mc{O}$ the set of quantum \emph{observables} on $(\Omega,\Sigma)$, i.e., the mappings $\Eo: \Sigma \rightarrow \lh$ that are positive, $\sigma$-additive and normalized so that $\Eo(\Omega)=\id_\hi$. 
The operators $\Eo(X)$, $X\in \Sigma$, in the range of an observable $\Eo$ are called \emph{effects}. 
If $\Eo(\{x_i\})\neq 0$ for some finite number of points $x_1,\ldots,x_N$ and $\sum_{i=1}^N \Eo(\{x_i\})=\id_\hi$, then we say that $\Eo$ has $N$-outcomes, or is \emph{$N$-valued}.
We often denote $\Eo(x) \equiv \Eo(\{x\})$ when there is no risk of confusion.

A special class of observables are those whose range consists of projections. 
We say that an observable $\Ao$ is \emph{sharp} if all its effects are projections, that is, $\Ao(X)^2 = \Ao(X)$ for every $X \in \Sigma$.

Every observable $\Eo$ has a \emph{Naimark dilation} into a sharp observable, i.e., there exist a Hilbert space $\hik$, a sharp observable $\Ao:\Sigma\to\lk$ and an isometric linear map $W:\hi\to\hik$ such that 
\begin{equation}
\Eo(X) = W^* \Ao(X) W
\end{equation}
for all $X\in\Sigma$; see e.g. \cite{CBMOA03}.
Since $W$ is an isometric linear map, $W^*W=\id_\hi$ and $WW^*$ is the projection onto $W(\hi)\subset\hik$.

We recall the following characterisation of sharp observables; see e.g.  \cite{LaYl04} for a proof.

\begin{proposition}\label{prel:sharp}
Let $\Eo$ be an observable and $(\hik,\Ao,W)$ its Naimark dilation. 
The following are equivalent:
\begin{itemize}
\item[(i)] $\Eo$ is a sharp observable.
\item[(ii)] $\Eo(X\cap Y) = \Eo(X)\Eo(Y)$ for all $X,Y\in\Sigma$.
\item[(iii)] $[\Ao(X),WW^*]=0$ for all $X\in\Sigma$. 
\end{itemize}
\end{proposition}

The set $\mc{O}$ is a convex set, and a convex combination $\lambda \Eo + (1-\lambda) \Fo$ of two observables $\Eo$ and $\Fo$ corresponds to their mixing. 
The extremal elements of $\mc{O}$, i.e., the observables that cannot be expressed as convex mixtures of other observables, are called \emph{extreme observables}.
All sharp observables are extreme, but it is known that there are also other extreme observables \cite{PSAQT82}.
For instance, in a finite $d$-dimensional Hilbert space we can construct an extreme observable with $N$ outcomes for each $N=d,\ldots,d^2$ \cite{HaHePe12}, but any sharp observable has at most $d$ outcomes.
A complete characterisation of extreme observables is given in \cite{Pellonpaa11}.

\subsection{Channels}

Quantum channels are devices performing quantum state transformations. 
Mathematically, a linear mapping $\Ec: \trh \rightarrow \trh$ is a quantum channel if it is \emph{completely positive} and \emph{trace-preserving} (Schr\"{o}dinger picture), or equivalently, if its dual map $\Ec^*: \lh \rightarrow \lh$ is completely positive and \emph{unital}, meaning that $\Ec^*(\id_\hi)=\id_\hi$ (Heisenberg picture). 

A unitary operator $U:\hi\to\hi$ defines a \emph{unitary channel} $\Uc$ by $\Uc (T) = UT U^*$ for all $T \in \trh$, or $\Uc^* (B) = U^* B U$ for all $B \in \lh$ in the Heisenberg picture.
The unitary channels are exactly the \emph{reversible channels}, i.e., those channels that have an inverse map which is a channel.
One of the important properties of unitary channels is that they are \emph{multiplicative}, namely
\begin{equation}\label{eq:unitary}
\Uc^* (BC) = U^* B U \, U^* C U = \Uc^* (B)\Uc^* (C)
\end{equation} 
for all $B,C \in \lh$. 

Every channel $\Ec$ has a \emph{Stinespring dilation}, i.e., there exist a Hilbert space $\hik$ and an isometric linear map $W:\hi\to\hi\otimes\hik$ such that 
\begin{equation}\label{prel:stines}
\Ec(T) = \ptr{WTW^*}
\end{equation}
for all $T\in\trh$, or equivalently
\begin{equation}
\Ec^*(B) = W^* B \otimes \id_\hik W
\end{equation}
for all $B\in\lh$; see e.g. \cite{CBMOA03}.

The following properties of unitary channels will be needed later.

\begin{proposition}\label{prel:unitary}
Let $\Uc^*:\lh\rightarrow\lh$ be a unitary channel and $(\hik,W)$ its Stinespring dilation. 
Then $\Uc^*$ satisfies the following equivalent conditions:
\begin{itemize}
\item[(i)] $\Uc^*$ is multiplicative.
\item[(ii)] $\Uc^*(P)$ is a projection for each projection $P\in\ph$.
\item[(iii)]  $[B\otimes \id_\hik, W W^*]=0$ for all $B\in \lh$.
\end{itemize}
\end{proposition}

\begin{proof}
We have already seen that $\Uc^*$ is multiplicative, so we only need to prove that the conditions (i)--(iii) are equivalent.
The implications (i)$\Rightarrow$(ii) and (iii)$\Rightarrow$(i) are straightforward to verify. 
In the following we prove that (ii)$\Rightarrow$(iii).

Assume (ii). For each $P \in \ph$, we then have
\begin{eqnarray}
W^* P \otimes \id_\K W &=& W^* P \otimes \id_\K W W^* P \otimes \id_\K W \nonumber \\
\Rightarrow W W^* P \otimes \id_\K W W^* &=& \left( W W^* P \otimes \id_\K W W^* \right)^2
\end{eqnarray}
implying that $W W^* P \otimes \id_\K W W^*$ is a projection. 
Since both $W W^*$ and $P \otimes \id_\K$ are projections, then by \cite[Lemma 2.2.1]{LaYl04} we have $\left[W W ^*, P \otimes \id_\K \right]=0$ for all $P \in \ph$. 
By the linearity and continuity of the mapping $B\mapsto \left[W W^*, B \otimes \id_\K \right]$ we get $\left[W W^*, B \otimes \id_\K \right]=0$ for all $B \in \lh$. 
\end{proof}

The set of all channels $\mc{C}$ is a convex set and the extremal elements of $\mc{C}$ are called \emph{extreme channels}. 
All multiplicative channels are extreme, but there exist also other extreme channels. 
For instance, fix a unit vector $\varphi\in\hi$ and define a channel by $\mc{E}^*(B)=\ip{\varphi}{B\varphi} \id_\hi$. 
This channel is clearly not multiplicative, but it is easy to verify that it is extreme.
The corresponding Schr\"{o}dinger channel $\Ec$ is known as a \emph{complete state space contraction}, since $\Ec(\rho) = P[\varphi]$ for all states $\rho \in \sh$.

\subsection{Measurement models}\label{sec:memo}

Every physical measurement is based on the same general concept: the observed system is brought into contact with some measuring apparatus and the value of the measured observable is read from the apparatus' pointer scale.

\begin{figure}
\includegraphics[width=0.9\textwidth]{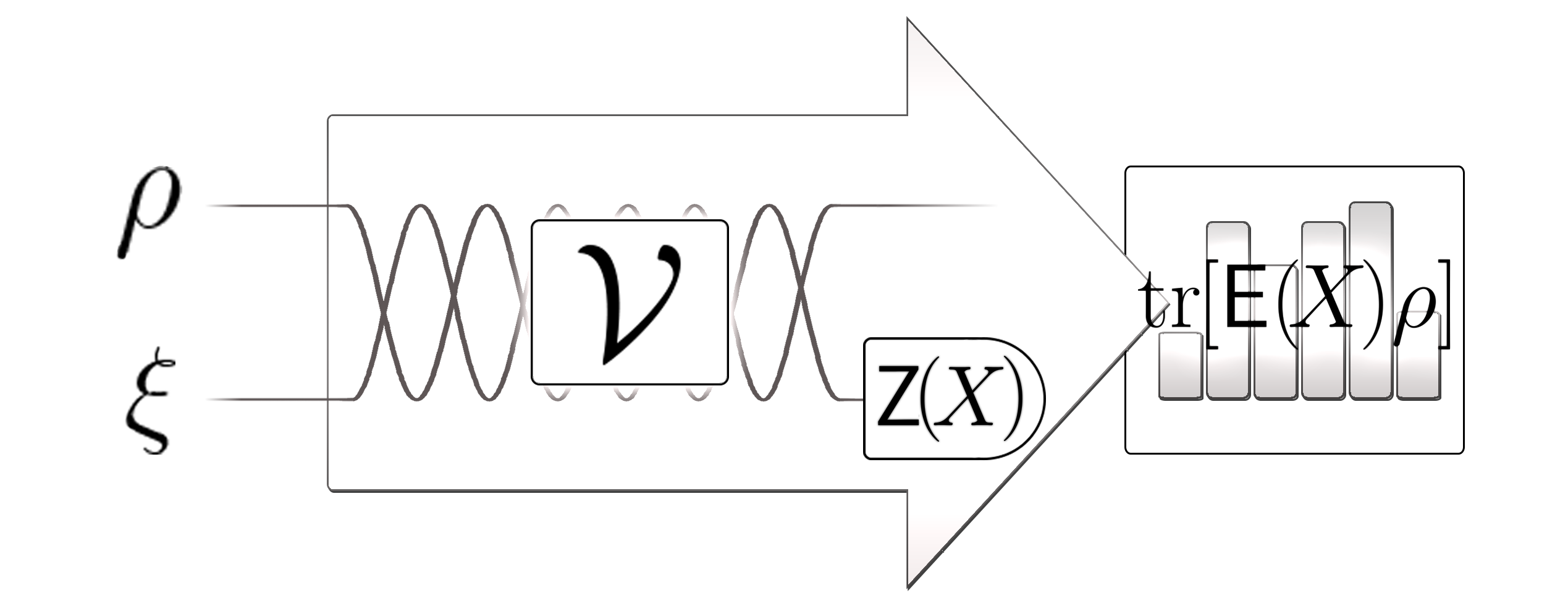} 
\caption{Measurement schematics: Initial state of the system ($\rho\in \sh$) is interacting with the initial probe state of the measurement apparatus $(\xi \in \sk)$ via measurement coupling ($\mc V: \mc T(\hi \otimes \K)\rightarrow \mc T(\hi \otimes \K)$). The measurement statistics of the pointer observable ($\Zo(X) \in \lk$) acting on the combined state ($\mc V(\rho\otimes\xi) \in \mc S (\hi \otimes \K)$) actualize the statistics of the measured observable $(\tr{\Eo(X) \, \rho})$. }\label{fig:mnt}
\end{figure}

Mathematically, a \emph{measurement model} is a 4-tuple $\langle \K, \Zo,\mc V, \xi \rangle$, where $\K$ is the Hilbert space associated to the apparatus, $\Zo: \Sigma \rightarrow \lk$ is the \emph{pointer observable}, $\mc V: \trhk \rightarrow \trhk$ is the channel describing the \emph{measurement interaction} between the system and the apparatus and $\xi \in \sk$ is the initial \emph{probe state}. 
The measurement statistics produce the measured observable $\Eo$ via the \emph{probability reproducibility condition} 
\begin{eqnarray}
\tr{\Eo(X) \, \rho} = \tr{\id_\hi \otimes \Zo(X) \, \mc V(\rho \otimes \xi) } \, ,  
\end{eqnarray}
which is required to hold for all initial system states $\rho \in \sh$ and $X\in \Sigma$; see Fig. \ref{fig:mnt}. In other words, the measured observable is
\begin{eqnarray}\label{obs:repro}
\Eo(X) = \ptr{\mc V^*\left(\id_\hi \otimes \Zo(X)\right) \, \id_\hi \otimes \xi} \, .
\end{eqnarray}

Following the terminology of \cite{OQP97}, we say that a measurement model is \emph{normal} when the pointer observable is sharp, the interaction is given by a unitary channel ($\mc V(\rho \otimes \xi) = G \rho \otimes \xi G^*$ for some unitary operator $G$ on $\hi\otimes\K$) and the initial probe state is pure ($\xi = P[\phi]$ for some unit vector $\phi \in \K$). In such a case we write the measurement model shortly as $\langle \K, \Zo, G, \phi \rangle$.
The measured observable induced by a normal measurement model $\langle \K, \Zo, G, \phi \rangle$ is 
\begin{equation}\label{eq:naimark-memo}
\Eo(X) = W_\phi^*\,  G^*(\id_\hi \otimes \Zo(X))G \, W_\phi \, , 
\end{equation} 
for all $X\in \Sigma$, where $W_\phi(\varphi)=\varphi\otimes \phi$ for all $\varphi\in\hi$. 
It is a fundamental result of quantum measurement theory that every observable has a normal measurement model \cite{Ozawa84}.
We say that \eqref{eq:naimark-memo} is a \emph{measurement dilation} of $\Eo$. 
It is obviously a special kind of Naimark dilation of $\Eo$. Using Prop. \ref{prel:sharp} one easily confirms that the sharpness of $\Eo$ is equivalent to $\left[ G^*(\id_\hi \otimes \Zo(X))G, W_\phi W_\phi ^* \right]=0$ for all $X\in \Sigma$.

Measurement processes can be used to model not only any observable but also an arbitrary quantum channel. 
Namely, a measurement model $\langle \K, \Zo, \Vc, \xi \rangle$ induces a channel 
\begin{equation}\label{chan:repro-s}
\Ec(\rho) = \ptr{\Vc(\rho \otimes \xi)} \, .
\end{equation} 
Note that the pointer observable plays no role in \eqref{chan:repro-s}.
The corresponding Heisenberg channel is given by
\begin{eqnarray}\label{chan:repro}
\Ec^*(B) = \ptr{\Vc^*(B\otimes \id_\K) \, \id_\hi \otimes \xi} \, .
\end{eqnarray}
As in the case of observables, every channel has a normal measurement model $\langle \K, \Zo, G, \phi \rangle$, and in that case \eqref{chan:repro} takes the form
\begin{equation}\label{eq:stinespring-2}
\Ec^*(B) = W_\phi^*\,  G^*(B \otimes \id_\K) G \,W_\phi \, , \quad B\in\lh \, . 
\end{equation} 
This is obviously a special kind of Stinespring dilation of $\Ec^*$ and by Prop. \ref{prel:unitary}, if $\mc E^*$ is a unitary channel, then $\left[G^*(B\otimes \id_\K)G, W_\phi W_\phi^* \right]=0$ for all $B \in \lh$.

We end this subsection with a simple observation related to the measurement models of extreme  observables and channels.
Suppose that an observable $\Eo$ has a measurement model $\langle \K, \Zo, \Vc, \xi \rangle$ where the probe state $\xi$ is mixed. 
We write the probe state $\xi$ as a convex decomposition $\xi = \sum_{i} \lambda_i P[\phi_i]$, and then 
\begin{eqnarray}
\Eo & = &  \ptr{\mc V^*\left(\id_\hi \otimes \Zo(X)\right) \, \id_\hi \otimes \xi} \\
&=& \sum_{i} \lambda_i  \, \ptr{\Vc^*(\id_\hi \otimes \Zo) \, \id_\hi \otimes P[\phi_i]} \nonumber \\
&=&  \sum_{i} \lambda_i \, \Eo_i \, , 
\end{eqnarray} 
where $\Eo_i$ are some observables. 
If $\Eo$ is extreme, then $\Eo_i=\Eo$ for each $i$,
and hence $\Eo$ has a measurement model $\langle \K, \Zo, \Vc, \phi_i \rangle$, where the probe state is pure.
An analogous argument is valid for the extreme channels.
We summarize this observation in the following proposition, earlier noted in \cite{DaPe05}.

\begin{proposition}\label{prop:ext}
Let $\mc D$ be an extreme observable/channel having a measurement model $\langle \K, \Zo, \mc V, \xi \rangle$. 
Then $\mc D$ has a measurement model $\langle \K, \Zo, \mc V, \xi' \rangle$, where  $\xi'$ is a pure state. 
\end{proposition}

\section{Programmable quantum multimeters}\label{sec:prog}

\subsection{Multimeters}

Generally speaking, \emph{multimeters} are measurement settings that can be programmed to implement any observable/channel from a specified set. 
The programming is done by changing some part of the measurement model: the probe state, the pointer observable, or the interaction.
Usually the easiest to physically realize, and thus the most interesting set of programmable multimeters, are those in which the initial probe state is changed.
In the rest of our investigation we reserve the term 'programming' for this state-programming scenario only.
Therefore, a multimeter is a measurement model in which we do not specify the pointer state, \emph{i.e.} a 3-tuple $\langle \K, \Zo, \Vc \rangle$.
This is a  \emph{normal multimeter} if $\Zo$ is sharp and $\Vc$ is unitary.

From an abstract point of view a quantum multimeter is a specific kind of function that maps quantum states into observables or channels.
This function must be physically realizable, meaning that it should be induced from a measurement process. 
The essential fact is that all functions from states to either channels or observables are not physically realizable in this way. 
For instance, it is easy to construct a bijective function between pure qubit states and qubit unitary channels. 
However, there is no programmable multimeter capable of implementing all qubit unitary channels \cite{NiCh97} nor all sharp qubit observables \cite{DaPe05}.

\subsection{Programming quantum observables}\label{sec:probs}

In this subsection we study the programmability of quantum observables. 
We denote the set of programmable observables with a given multimeter $\mc{M}$ by $\mc{O}_\mc{M}$. 
The analysis done in articles \cite{DuBu02,FiDuFi02,DaPe05} indicates that $\mc{O}_\mc{M}$ is always a proper subset of $\mc{O}$, meaning that a fixed multimeter cannot be used to program all observables.
Here we elaborate this result by showing that different sharp observables demand orthogonal programming vectors, proving that $\mc{O}_\mc{M}$ can contain at most $\dim\K$ sharp observables.
Note that due to Proposition \ref{prop:ext} the programming protocols of all extreme observables can be restricted to vector states. 

\begin{proposition}\label{obs:programming} 
Let $\mc{M}=\langle \hik, \Zo, \Vc \rangle$ be a multimeter and suppose that two different sharp observables $\Ao_1$ and $\Ao_2$ can be programmed with vector states $\phi_1$ and $\phi_2$, respectively. 
Then $\phi_1$ and $\phi_2$ are orthogonal.
\end{proposition}

\begin{proof}
We fix disjoint sets $X,Y\in\Sigma$ such that $\Ao_1(X) \Ao_2(Y) \neq 0$.
These kind of sets exist, for if $\Ao_1(X) \Ao_2(Y)=0$ holds for all disjoint sets $X,Y\in\Sigma$, then
\begin{eqnarray}
\Ao_1(X) = \Ao_1(X) \Ao_2(\Omega) = \Ao_1(X) \Ao_2(X) = \Ao_1(\Omega) \Ao_2(X) = \Ao_2(X) \, , 
\end{eqnarray}
contradicting $\Ao_1\neq\Ao_2$.

The condition $\Ao_1(X) \Ao_2(Y) \neq  0$ means that the projections $\Ao_1(X)$ and $\Ao_2(Y)$ are non-orthogonal. 
Hence, there exist unit vectors $\varphi_1$ and $\varphi_2$ in $\hi$ such that $\Ao_1(X)\varphi_1=\varphi_1$, $\Ao_2(Y)\varphi_2=\varphi_2$ and $\ip{\varphi_1}{\varphi_2}\neq 0$.
Since $X$ and $Y$ are disjoint, we have $\Ao_1(Y)\varphi_1= 0$ and $\Ao_2(X)\varphi_2= 0$.

Denote $\Fo = \Vc^*(\id_\hi \otimes \Zo)$. 
Then
\begin{eqnarray}
\ip{\varphi_1 \otimes \phi_1}{\Fo(Y) \varphi_1 \otimes \phi_1} = \ip{\varphi_1}{\Ao_1(Y) \varphi_1} = 0
\end{eqnarray}
and
\begin{eqnarray}
\ip{\varphi_2 \otimes \phi_2}{\Fo(Y) \varphi_2 \otimes \phi_2} = \ip{\varphi_2}{\Ao_2(Y) \varphi_2} =1,
\end{eqnarray}
from which it follows that $\Fo(Y) \varphi_1 \otimes \phi_1 = 0$ and $\Fo(Y) \varphi_2 \otimes \phi_2 =  \varphi_2 \otimes \phi_2$.
Using these two equations we obtain
\begin{eqnarray}
0 &=& \ip{\Fo(Y) \varphi_1 \otimes \phi_1}{\varphi_2 \otimes \phi_2} = \ip{\varphi_1 \otimes \phi_1}{\Fo(Y) \varphi_2 \otimes \phi_2} \nonumber \\
&=& \ip{\varphi_1 \otimes \phi_1}{\varphi_2 \otimes \phi_2} = \ip{\varphi_1}{\varphi_2}\ip{\phi_1}{\phi_2} \, , 
\end{eqnarray}
proving that $\ip{\phi_1}{\phi_2}=0$. 
\end{proof}

A stronger version of Prop. \ref{obs:programming} holds when the multimeter is assumed to be normal.

\begin{proposition}\label{obs:ex} Let $\mathcal{M}= \langle \K, \Zo, G \rangle$ be a normal multimeter and suppose two different observables, a sharp observable $\Ao$ and an extreme observable $\Eo$, can be programmed with vector states $\phi$ and $\phi'$, respectively. 
Then $\phi$ and $\phi'$ are orthogonal.
\end{proposition}

\begin{proof}
Assume $\ip{\phi}{\phi'}\neq 0$. 
Then $\phi'$ can be written as $\phi'=\alpha \phi + \beta \eta$, where $\ip{\phi}{\eta}=0$ and $\alpha,\beta \in \C\setminus\{0\}$ satisfy $|\alpha|^2 + |\beta|^2 = 1$. 
We have
\begin{equation}
W_{\phi'}= \alpha W_{\phi} + \beta W_{\eta} \, , \quad W^*_{\phi'}= \bar{\alpha} W^*_{\phi} + \bar{\beta} W^*_{\eta} \, , 
\end{equation}
and
\begin{equation}
W_{\phi}^* W_{\eta}= \ip{\phi}{\eta} \id =0 = W_{\eta}^*  W_{\phi} \, .
\end{equation}

Denote $\Zo'(X)=G^*\left(\id_\hi\otimes\Zo(X)\right) G$ for all $X\in\Sigma$. 
We recall from Prop. \ref{prel:sharp} that the sharpness of $\Ao$ is equivalent to $[\Zo'(X), W_{\phi} W_{\phi}^*]=0$ for all $X\in\Sigma$.
Therefore, 
\begin{equation}
W_{\phi}^* \Zo'(X) W_{\eta} = W^*_{\phi}W_{\phi} \left( W_{\phi}^* \Zo'(X) W_{\eta} \right) = W^*_{\phi} \Zo'(X) W_{\phi}  W_{\phi}^*W_{\eta} = 0 \, .
\end{equation}
It follows that
\begin{eqnarray} 
\Eo(X) &=& W_{\phi'}^* \Zo'(X) W_{\phi'} \nonumber \\
&=& |\alpha|^2 W_{\phi}^* \Zo'(X) W_{\phi} + \bar{\alpha} \beta W_{\phi}^* \Zo'(X) W_{\eta} \nonumber \\ 
& &+\alpha \bar{\beta} W_{\eta}^* \Zo'(X) W_{\phi}+ |\beta|^2 W_{\eta}^* \Zo'(X) W_{\eta} \nonumber \\
&=& |\alpha|^2 \Ao(X) + |\beta|^2 \Fo(X),
\end{eqnarray}
where $\Fo$ is an observable induced by the measurement $\langle \K, \Zo, G, \eta \rangle$. 
This contradicts the extremality of $\Eo$, hence $\ip{\phi}{\phi'} \neq 0$ must be false. \qed
\end{proof}

One may wonder if two extreme observables require orthogonal program states. 
This is not the case as the following example illustrates.

\begin{example}
Covariant phase space observables have an important role in quantum mechanics \cite{CSWTG00}. 
It is well known that every covariant phase space observable rises as an operator density defined by the \emph{Weyl operators} $W(q,p)=e^{\frac{i q p}{2}} e^{-i q P} e^{i p Q}$, $q,p\in \R$ and a state $\xi \in \sk$ in the form
\begin{eqnarray} \label{density}
\Eo^\xi(Z)=\frac{1}{2\pi} \int_Z W(q,p)\, \xi \, W(q,p)^* dq \, dp,
\end{eqnarray} for all $Z \in \mc B(\R^2)$. 
It is easy to see that $\Eo^\xi$ is extremal in the set of all covariant phase space observables if and only if $\xi=P[\phi]$ for some unit vector $\phi\in \K$, and furthermore it has been shown in \cite{HePe12} that $\Eo^\xi$ is extremal in the set of all observables on $\R^2$ if and only if $\ip{\phi}{W(q,p)\phi}\neq 0$ for all $q,p\in \R$. 
It follows that, for example, every Gaussian state $\phi_{a,b}(x)=N \text{exp}(-\frac{a}{2} x^2 + b x )$, where $a>0$, $b \in \R$ and $N$ is an appropriate normalizing constant, induces some extremal covariant phase space observable via (\ref{density}).

A physically feasible normal measurement model for such observables, based on \emph{eight port homodyne detection}, is described in \cite{KiLa07}. 
In this particular model the vector state $P[\phi]$ of $\Eo^{P[\phi]}$ is regarded as the initial probe state. The above considerations show that two extremal observables do not necessarily require orthogonal vector states for programming since for example $\phi_{a,b}$ and $\phi_{c,d}$ are non-orthogonal for all $a,c>0$, $b,d\in \R$.
\end{example}

In the statement of Prop. \ref{obs:ex}, one can require that $\Eo$ is merely an extreme element of $\mc{O}_{\mc{M}}$ rather than of $\mc{O}$.
The proof is still valid without any changes. 
Therefore, Prop. \ref{obs:ex} gives some indication on what the set $\mc{O}_{\mc{M}}$ looks like. 
In the case where the maximal amount of sharp observables can be programmed, the structure of $\mc{O}_{\mc{M}}$  is particularly simple.

\begin{proposition}\label{obs:basis}
Let $\dim \K<\infty$.
Suppose one can program $n=\dim \K$ sharp observables $\Ao_1,\ldots,\Ao_n$, with a normal multimeter $\mathcal{M}$. 
Then $\mc{O}_{\mc{M}}$ is the convex hull of the set $\{\Ao_1,\ldots,\Ao_n \}$. 
\end{proposition}

\begin{proof}
The programming vectors $\phi_1,\ldots,\phi_n$ form an orthonormal basis of $\hik$. 
Every unit vector $\psi\in\hik$ has a basis expansion $\psi = \sum_{i=1}^n \ip{\phi_i}{\psi} \phi_i$.
A similar calculation as in the proof of Prop. \ref{obs:ex} shows that the observable $\Eo$ obtained by using $\psi$ as a program state is
\begin{equation}
\Eo = \sum_{i=1}^n \mo{\ip{\phi_i}{\psi}}^2 \Ao_i \, .
\end{equation}
hence a mixture of the observables $\Ao_1,\ldots,\Ao_n$.

Every mixed program state $\xi$ has a convex decomposition into pure states, $\xi = \sum_{i} \lambda_i P[\psi_i]$.
The resulting observable $\Eo$ is a mixture of the observables corresponding to the pure states $P[\psi_i]$.
Thus, every observable in $\mc{O}_{\mc{M}}$ is a mixture of the observables $\Ao_1,\ldots,\Ao_n$.
\end{proof}


\subsection{Programming quantum channels}

As in the case of programming observables there is no multimeter which could be deterministically programmed to realize all channels, meaning that the set of programmable channels $\mc{C}_{\mc{M}}$ with a given multimeter $\mc{M}$ is always a proper subset of $\mc{C}$. 
A special case in which the multimeter is assumed to be normal was first analysed by Nielsen and Chuang who showed that, when such a multimeter is programmed to realize two different unitary channels, the programming vectors must be orthogonal \cite{NiCh97}.

In this subsection we generalize this result for general (non-normal) multimeters capable of realizing two different unitary channels.
Again Proposition \ref{prop:ext}, together with the extremality of unitary channels, shows that one only needs to consider pure programming states. 
We remind that, since the pointer observable plays no role in the channel programming scenarios, the corresponding entry in the measurement model is simply omitted; the programmable multimeter is given by a pair $\langle \K, \mc V \rangle$. 
Often in the literature this also known as programmable \emph{gate array} or  \emph{processor}.

\begin{proposition}\label{chan:programming} 
Let $\mathcal{M}=\langle \hik, \Vc \rangle$ be a multimeter and suppose that two different unitary channels $\Uc_1$ and $\Uc_2$ can be programmed with vector states $\phi_1$ and $\phi_2$, respectively. 
Then $\phi_1$ and $\phi_2$ are orthogonal.
\end{proposition}

\begin{proof}
Suppose $\ip{\phi_1}{\phi_2}\neq 0$. 
Since $\Uc_1 \neq \Uc_2$, there exists a projection $P\in\ph$ such that $\Uc_1(P) \neq \Uc_2(P)$. 
We fix a sharp observable $\Ao$ such that $\Ao(X)=P$ for some $X\in\Sigma$.
For instance, $\Ao$ can be a two-valued observable consisting of projections $P$ and $\id_\hi-P$.
Further, we fix a measurement model $\langle \K_0, \Zo, \Vc_0, \eta \rangle$ for $\Ao$, where the program state $\eta$ can be chosen pure. 
Then the concatenated multimeter in Fig. \ref{fig:concatenation} implements different sharp observables $\Ao_1=\Uc_1^*(\Ao)$ and $\Ao_2=\Uc_2^*(\Ao)$ with non-orthogonal programming vectors $\phi_1\otimes \eta$ and $\phi_2\otimes \eta$, respectively.
This contradicts Prop. \ref{obs:programming}, and hence $\ip{\phi_1}{\phi_2}\neq 0$ must be false.
\end{proof}

\begin{figure}
\includegraphics[width=0.95\textwidth]{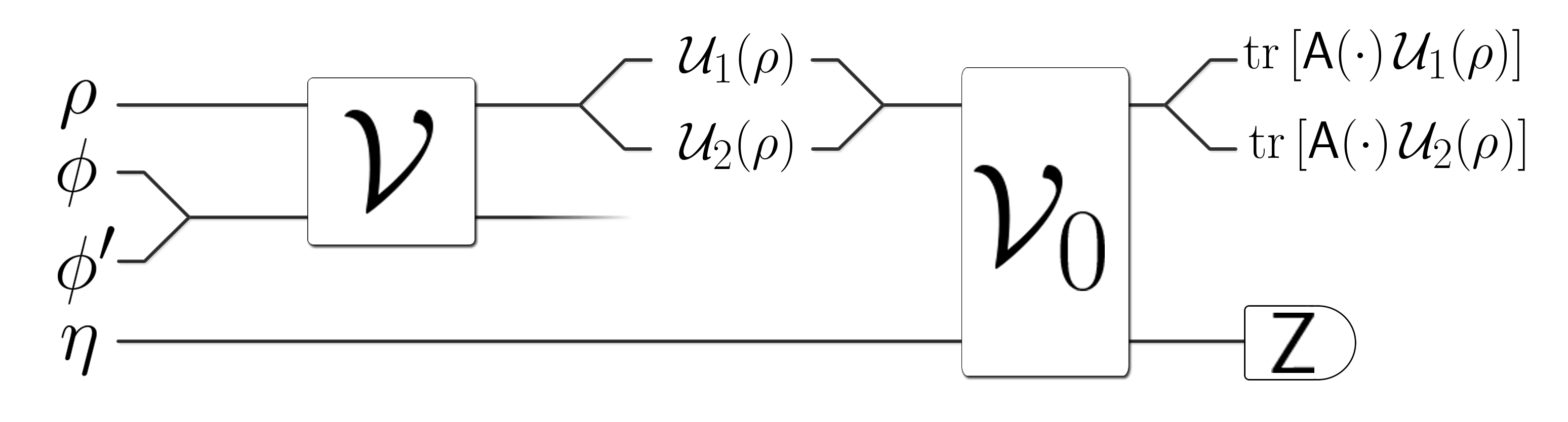} 
\caption{An illustration of a concatenation of a multimeter $\langle \K, \mc V \rangle$ and an $\Ao$-measurement $\langle \K_0, \Zo, \mc V_0, \eta \rangle$ used in proposition  \ref{chan:programming}.}
\label{fig:concatenation}
\end{figure}

As in the case of observables, a stronger version of Prop. \ref{chan:programming} is obtained when the multimeter is assumed to be normal.
The proof is analogous to that of Prop. \ref{obs:ex}.

\begin{proposition}\label{chan:ex}
Let $\mathcal{M}= \langle \hik, G \rangle$ be a normal multimeter and suppose two different channels, a unitary channel $\Uc$ and an extreme channel $\Ec$, can be programmed with vector states $\phi$ and $\phi'$, respectively. 
Then $\phi$ and $\phi'$ are orthogonal.
\end{proposition}

\begin{proof}
Assume $\ip{\phi}{\phi'}\neq 0$. 
Then $\phi'$ can be written as $\phi'=\alpha \phi + \beta \eta$, where $\ip{\phi}{\eta}=0$ and $\alpha,\beta \in \C\setminus\{0\}$ satisfy $|\alpha|^2 + |\beta|^2 = 1$. 
We recall from Prop. \ref{prel:unitary} that the unitarity of the channel $\Uc$ implies $[G^* (B\otimes \id_\hik)G, W_\phi W_\phi^*]=0$ for all $B\in \lh$.
It follows that 
\begin{equation}
W_{\phi}^*  G^*(B\otimes \id_\hik)G W_{\eta} =  0 \, ,
\end{equation}
and hence
\begin{eqnarray} 
\Ec^*(B) &=& W_{\phi'}^*  G^*(B \otimes \id_\hik)G  W_{\phi'} \nonumber \\
&=& |\alpha|^2 \Uc^*(B) + |\beta|^2 \Fc^*(B) \, , 
\end{eqnarray}
where $\Fc$ is a channel. 
This contradicts the extremality of $\Ec$, thus $\ip{\phi}{\phi'} \neq 0$ must be false.
\end{proof}

The next example shows that sheer extremality of the channels is not enough to imply orthogonality of the programming states.

\begin{example}
Choose two non-orthogonal unit vectors $\varphi,\varphi'\in\hi$. 
The corresponding complete state space contractions are extreme channels, and they can be programmed with a normal multimeter $\langle \hi, G_{\text{SWAP}}\rangle$ and non-orthogonal program states $\varphi$ and $\varphi'$, respectively.
Here $G_{\text{SWAP}}$ is the 'swap' unitary operator defined as
\begin{equation*}
G_{\text{SWAP}} \psi_1\otimes\psi_2 = \psi_2 \otimes \psi_1 \, .
\end{equation*}
\end{example}

The proof of the following result is the same as the proof of Prop. \ref{obs:basis} with obvious changes, so we omit it.

\begin{proposition} 
Let $\dim \K<\infty$.
Suppose one can program $\dim \K=n$ unitary channels $\Uc_1,\ldots,\Uc_n$, with a normal multimeter $\mathcal{M}$. 
Then $\mc{C}_\mc{M}$ is the convex hull of the set $\{\Uc_1,\ldots,\Uc_n \}$. 
\end{proposition}

\begin{example}\label{ex:push} An example of a normal multimeter in previous proposition is the "push-a-button" multimeter for which $G = \sum_{i=1}^n U_i \otimes P[\phi_i]$. Indeed we have $\mc U_i(\rho) = \ptr{G (\rho \otimes P[\phi_i]) G^*}$ for every $i=1,...,n$. We shall return to this multimeter in the last section \ref{sec:eff}.
\end{example}

\subsection{Post-processing assisted programming of quantum observables}

\emph{Post-processing} is classical information processing of the obtained measurement statistics: merging together, relabeling and weighting measurement outcomes in a stochastic manner. 
Mathematically post-processing is described by a classical-to-classical channel $\Lambda$ between probability measures on measurable spaces $(\Omega, \Sigma)$ and $(\Omega', \Sigma')$. 

\begin{figure}[h]
\includegraphics[width=1.0\textwidth]{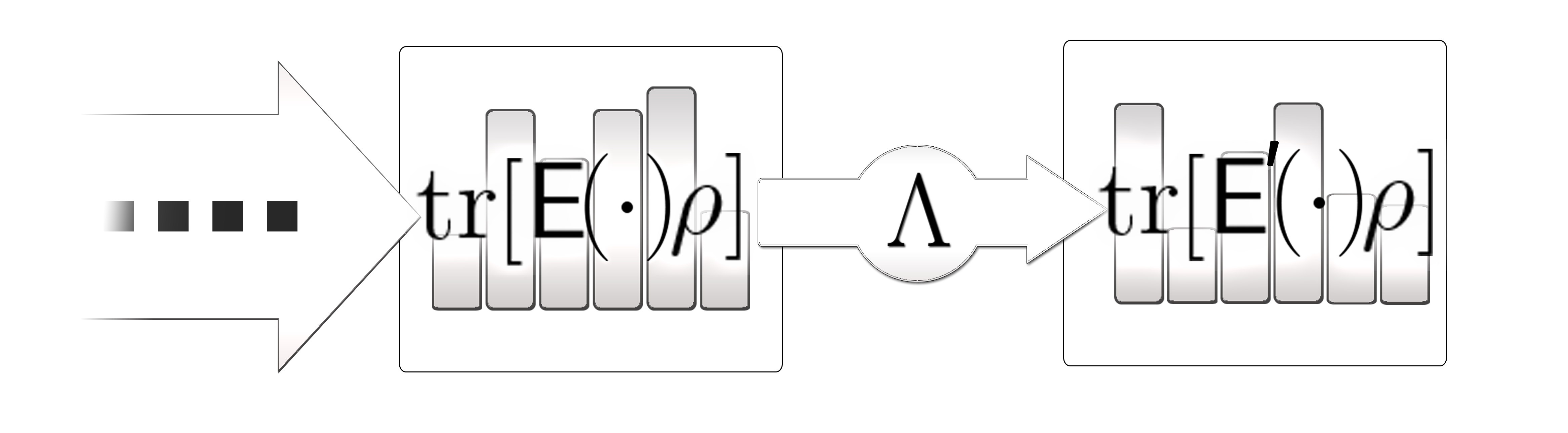} 
\caption{Post-processing is a classical manipulation of measured statistics described by a classical-to-classical channel $\Lambda$. }\label{fig:postpr}
\end{figure}

In particular this channel $\Lambda$ maps measurement statistics of the measured observable $\Eo:\Sigma\rightarrow\lh$ into those of another \emph{fuzzy} observable $\Eo':\Sigma'\rightarrow\lh$ via $\Lambda(\tr{\Eo(\cdot) \, \rho}) = \tr{\Eo'(\cdot) \, \rho}$, for all $\rho \in \sh$; see Fig. \ref{fig:postpr}. 
Under certain technical assumptions (see \cite{JePuVi08} for details), the action of the channel $\Lambda$ can be equivalently described by a \emph{Markov kernel}, i.e., a map $k: (\Omega,\Sigma') \rightarrow \left[0,1\right]$ such that
\begin{itemize}
\item[(i)] $k(x, \cdot)$ is a probability measure for every $x\in \Omega$,
\item[(ii)] $k(\cdot, X)$ is measurable for every $X \in \Sigma'$,
\end{itemize}
acting via formula $\Eo'(X) = \int k(x, X) \, d\Eo(x)$. 
The properties of post-processing have been studied for example in \cite{MaMu90a,BuDaKePeWe05,Heinonen05,JePu07}.
 
Post-processing can be taken as an additional component of a measurement model, so that a measurement model is understood as a 5-tuple $\langle \K, \Zo, \mc V, \xi,k \rangle$.  
If post-processing is included in programming, then the programmability of quantum observables changes. 
For instance, jointly measurable observables can be implemented without changing the probe state at all; we just fix a measurement model for their joint observable and post-process the measurement statistics by calculating the marginals.  

Surprisingly, different sharp observables may not require orthogonal programming states if post-processing is allowed. 
In the following we present an example of a measurement model capable of programming three sharp spin-observables, $\So_i(\pm) = \frac{1}{2}(\id_{\C^2} \pm \sigma_i)$, $i=1,2,3$, with non-orthogonal programming vectors. Here $\sigma_i$, $i=1,2,3$, are the conventional \emph{Pauli spin operators}.
The example is a more explicit version of the one constructed in \cite{ZiBu05}. 

\begin{example}
Let the apparatus Hilbert space be $\K= \C^4$ with orthonormal basis $\{|i \rangle \, | \, i=0,1,2,3\}$ and define the (non-orthogonal) programming vectors $\phi_i = \frac{1}{\sqrt{2}}(|0 \rangle+|i \rangle)$, $i=1,2,3$. One easily checks that $G = \sum_{j,k=0}^3 \frac{1}{2} \sigma_j \sigma_k \sigma_j \otimes |j \rangle \langle k |$, where $\sigma_0 = \id_{\C^2}$, is an unitary operator on $\C^2 \otimes \C^4$. 
A straightforward calculation shows that the effects of the measured observable in the model $\langle \K, \Zo, G, \phi_i\rangle$, where $\Zo(j)=|j\rangle\langle j |$, are given by $\Eo_i(j) = \frac{1}{4}\left( \id_{\C^2} + \sigma_j \sigma_i \sigma_j\right)$. Using three different post-processings we get
\begin{eqnarray}
\begin{array}{lr}
\So_1(+) = \Eo_1(0)+\Eo_1(1)\, , & \So_1(-) = \Eo_1(2)+\Eo_1(3) \\
\So_2(+) = \Eo_2(0)+\Eo_2(2)\, , & \So_2(-) = \Eo_2(1)+\Eo_2(3) \\
\So_3(+) = \Eo_3(0)+\Eo_3(3)\, , & \So_3(-) = \Eo_3(1)+\Eo_3(2)
\end{array}.
\end{eqnarray}

More generally, define a unit vector $\phi_{(\alpha, \vec{a})} \in \K$ to be $\phi_{(\alpha, \vec{a})} = \left(\alpha |0\rangle +\sum_{i=1}^3 a_i |i\rangle \right)$, where $\alpha \in \R$ and $\vec{a}=(a_1,a_2,a_3)\in\R^3$ satisfy the normalization $\alpha^2 + \sum_i a_i^2 = 1$. Programming of the above multimeter with $\phi_{(\alpha, \vec{a})}$ induces an observable $\Eo_{(\alpha, \vec{a})}(j) = \frac{1}{4}\left(\id_{\C^2} + 2 \alpha\, \sigma_j (\vec{a} \cdot \vec{\sigma}) \sigma_j\right),$ which is easily proven using the property $\{\sigma_j, \sigma_k\} = 2\delta_{jk} \id_{\C^2}$ for $j,k =1,2,3$. In particular the choise $\alpha = 1/\sqrt{2}, a_i = 1/\sqrt{6},$ $i=1,2,3$ induces a \emph{symmetric informationally complete} observable $\Eo(j)= \frac{1}{4}(\id_{\C^2} + \frac{1}{\sqrt{3}} \sum_i \sigma_j \sigma_i \sigma_j )$.
\end{example}

The previous example illustrates that some sharp observables can indeed be programmed with non-orthogonal programming vectors if post-processing is allowed. 
The explanation for this is simple: although post-processing should be viewed as classical processing of information, it effectively alters the pointer observable. 
To see this, let $\langle \K, \Zo, \mc V, \phi, k \rangle$ be a measurement model for $\Eo$, where $k: (\Omega, \Sigma) \rightarrow \left[0,1\right]$ is a Markov kernel. 
Using the formula \eqref{obs:repro}, linearity and continuity we have
\begin{eqnarray}
\Eo(X) &=& \int k(x, X) \, \ptr{ \mc V^*( \id_\hi \otimes d\Zo(x)) \, \id_\hi \otimes P[\phi]}\nonumber \\
&=&\ptr{ \mc V^*( \id_\hi \otimes \left(\int k(x, X)\, d\Zo(x)\right)) \id_\hi \otimes P[\phi]} \nonumber \\
&=& \ptr{ \mc V^*( \id_\hi \otimes \Zo'(X) ) \, \id_\hi \otimes P[\phi]},
\end{eqnarray} 
where  $\Zo'(X) = \int k(x, X) \, d\Zo(x)$. 
This shows that $\langle \K, \Zo,\mc  V, \phi, k \rangle$ and $\langle \K, \Zo', \mc V, \phi \rangle$ are equivalent as measurements of $\Eo$. 
But because the pointer is no longer kept fixed (if the Markov kernel is changed), the derivation in Prop. \ref{obs:programming} leading to the orthogonality of programming vectors does not hold anymore.

We note that if the pointer observable is altered arbitrarily, then all observables acting on a given Hilbert space can be implemented with a single a measurement model \cite{HeMiRe14}.
It is, however, left as a open question whether there exists a universal state-programmable multimeter capable of measuring every observable when post-processing is allowed.

Since the post-processing only affects the pointer observable, it is obvious that the programmability of channels does not change even if post-processing is allowed.

\section{Efficiency of a Quantum Multimeter}\label{sec:eff}

In this section we shortly study the effectiveness of quantum multimeters and give limits to efficient programming. Throughout the section we assume that post-processing is not allowed.

As discussed in Sec. \ref{sec:memo}, a normal measurement model $\langle \K, \Zo, G, \phi \rangle$ constitutes a special type of Naimark dilation of the measured observable $\Eo$, called a measurement dilation of $\Eo$.
Different measurement dilations of a given observable can be arranged by the dimensions of the auxiliary Hilbert spaces. 
This gives us a natural quantification of effectiveness: \emph{the smaller} the apparatus' dimension, \emph{the more efficient} the measurement.

It has been shown in \cite{Pellonpaa14} that the minimal dimension of the apparatus $\K$ corresponding to the minimal measurement dilation of a $N$-valued (sharp) observable is $\dim \K =N$ and such a minimal measurement can always be found. 
For readers convenience we present here a slightly different proposition proving this fact.

\begin{proposition}\label{prop:minimal} 
Let $\Ao$ be a $N$-valued sharp observable.
\begin{itemize}
\item[(a)] There exists a normal measurement model $\langle \K, \Zo, G, \phi \rangle$ of $\Ao$ with $\dim \hik  = N$.
\item[(b)] Any measurement model $\langle \K, \Zo, \mc V, \xi \rangle$ of $\Ao$ satisfies $\dim \hik  \geq N$.
\end{itemize}
\end{proposition}

\begin{proof} 
\begin{itemize}
\item[(a)] 
Let $\K$ be a Hilbert space with $\dim \K = N$ and fix an orthonormal basis $\{\varphi_j \}_{j=1}^N$ for $\hik$. 
For each $j=1,...,N$, define a unitary operator $G_j$ on $\hik$ as
\begin{eqnarray}
G_j \varphi_k = \left\lbrace 
\begin{array}{ll}
\varphi_1, & \quad j=k \\
\varphi_k, & \quad j\neq k\neq 1\\
\varphi_j, & \quad k=1
\end{array} \right. .
\end{eqnarray}
Further, define $G = \sum_j \Ao(j) \otimes G_j^*$, which is a unitary operator on $\hi\otimes\hik$. 
By choosing the pointer observable as $\Zo(k) = |\varphi_k\rangle\langle\varphi_k|$ and the initial probe state as $\phi=\varphi_1$, we get
\begin{eqnarray}
W^*_{\phi} \, G^*\left( \id_\hi \otimes \Zo(k)\right) G \, W_{\phi} &=& \sum_j | \ip{G_j \varphi_k}{\varphi_1}|^2 \, \Ao(j)\nonumber \\
&=& \Ao(k) \, , 
\end{eqnarray}
proving that $\langle \K, \Zo, G, \phi \rangle$ is a measurement model of $\Ao$. 
\item[(b)]
Let $\langle \K, \Zo, \mc V, \xi \rangle$ be a measurement model of $\Ao$. 
We first note that the pointer observable $\Zo$ must be $N$-valued (though not necessarily sharp) and the probe state can be chosen pure, $\xi = |\phi\rangle\langle\phi|$.
Since $\Ao$ is sharp, there exist orthonormal vectors $\varphi_1,\ldots,\varphi_N$ such that $\Ao(j) \varphi_i = \delta_{ij} \varphi_j$. 
Then the probability reproducibility condition gives
\begin{eqnarray}
\delta_{ij} = \tr{\Ao(j)\, P[\varphi_i]} &=& \tr{\id_\hi \otimes \Zo(j) \, \mc V (P[\varphi_i]\otimes \xi) } \nonumber \\
&=&\tr{\Zo(j) \, \hptr{ \mc V(P[\varphi_i]\otimes\xi)}} \nonumber \\
&=&\tr{ \Zo(j) \, \zeta_i},
\end{eqnarray}
where $\zeta_i = \hptr{\mc V(P[\varphi_i]\otimes \xi)}$.
It follows that $\Zo(j) \zeta_i = \delta_{ij} \, \zeta_j$, and therefore
\begin{eqnarray}
\zeta_i  \, \zeta_j = \zeta_i  \, \Zo(j) \zeta_j = (\Zo(j) \zeta_i)^* \, \zeta_j = \delta_{ij}  \, \zeta_j^2,
\end{eqnarray}
meaning that the supports of the operators $\zeta_i$ and $\zeta_j$ are orthogonal as subspaces whenever $i\neq j$. 
We thus have $N$ orthogonal subspaces in $\hik$, so $\dim\hik \geq N$. 
\end{itemize}
\end{proof}

Obviously this sets further limitations to the programming scenarios of sharp observables in addition to those seen in \ref{sec:probs}. We summarize these observations in the following proposition.

\begin{proposition}
Programming of different sharp observables $\Ao_1,\ldots,\Ao_n$, with $N_i$ values, respectively, can be done only if the apparatus $\K$ satisfies $\dim \K \geq \max \{ n, N_1,...,N_n \}.$
\end{proposition}

\begin{example}
Let $\Ao_1$ and $\Ao_2$ be sharp spin-observables. 
Since the effects $\Ao_1(\pm)$ and $\Ao_2(\pm)$ are rank-1 projections in a two dimensional Hilbert space, there exists a unitary operator $R$ on $\C^2$ such that  $\Ao_2 =  R^* \Ao_1 R$.
Now let $\phi_1$ and $\phi_2$ be any pair of orthonormal vectors and define a unitary operator $G = G_{\text{SWAP}} (\id_{\C^2} \otimes P[\phi_1] + R \otimes P[\phi_2])$. 
One easily verifies that $\langle \C^2, \Ao_1, G\rangle$ is a multimeter that can be programmed to measure $\Ao_1$ and $\Ao_2$ with vectors $\phi_1$ and $\phi_2$, respectively.
This shows that at least in this particular case the minimal programming dimension can indeed be reached. 
\end{example}
 \begin{figure}[h]\label{fig:push}
\includegraphics[width=1.0\textwidth]{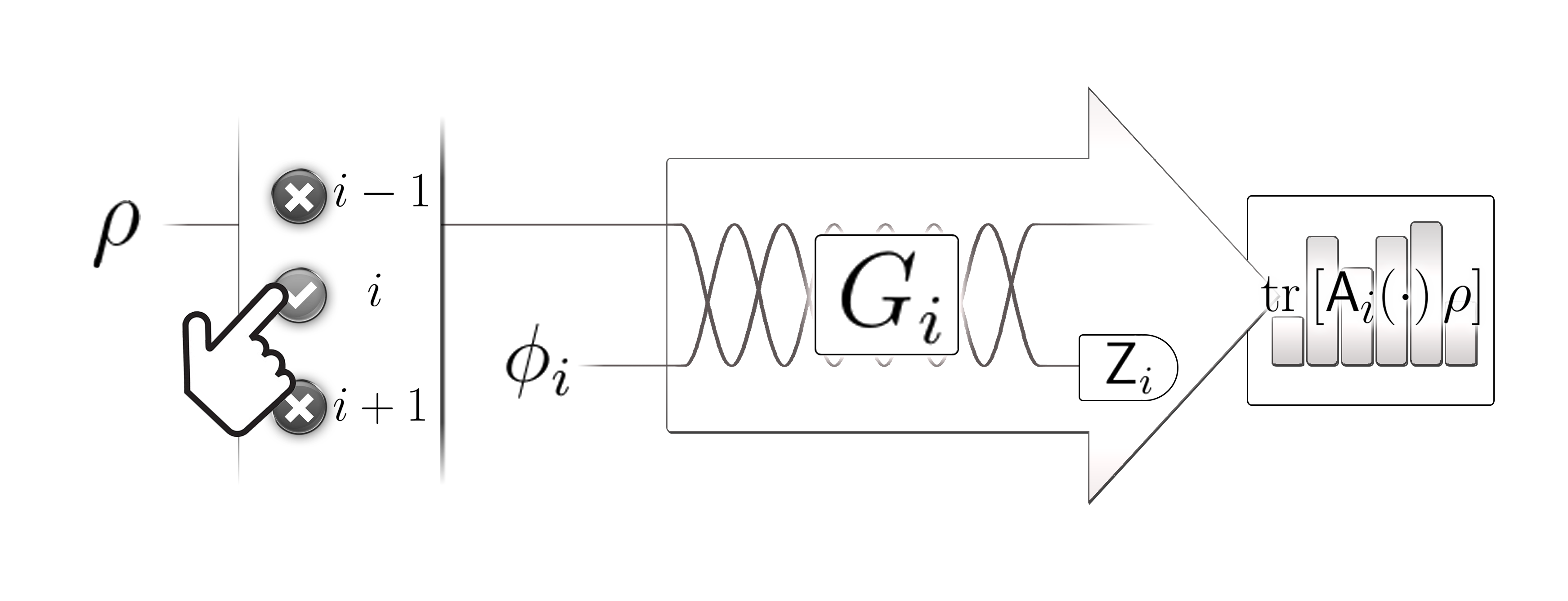} 
\caption{Illustration of the "push-a-button" programming protocol. Minimal measurement devices for each observable are bundled together and one simply chooses beforehand which one to use.}
\end{figure}
An inefficient way of programming $n$ observables is to have minimal measurements $\langle \K_j, \Zo_j, G_j, \phi_j\rangle$ for each observable $\Ao_j$ individually and simply to select beforehand which one should be performed (see Fig. \ref{fig:push}) .
Mathematically this ''push-a-button'' protocol is described by a programmable multimeter $\langle \K, \Zo, G \rangle$ and programming vectors $\Phi_i, \ i=1,...n$ where 
\begin{eqnarray*}
\K=\otimes_{j=1}^n \K_j  \otimes \C^n \\
\Zo=\left( \otimes_{j=1}^n \Zo_{j} \right)\otimes \id_{\C^n} \\
G=\left( \sum_{i=1}^n G_i \otimes_{j\neq i}^n \id_j   \otimes |i\rangle\langle i | \right) \\
\Phi_i = \left( \otimes_{j=1}^n \phi_j \right) \otimes |i\rangle
\end{eqnarray*}

The size of this "push-a-button" -multimeter is $\dim \K = n \cdot \Pi_{i=1}^n N_i$. We take this trivial programming scenario, which nevertheless can always be performed for any finite set of (sharp) observables, as the upper bound of effectiveness of a multimeter. 
We conclude that in every interesting programming protocol capable of programming sharp observables $\Ao_i$, $i=1,...,n$, the size of the apparatus $\dim \K$ satisfies the bounds
\begin{eqnarray}
\max\{n,N_1,...,N_n\} \leq \dim \K \leq n \cdot \Pi_{i=1}^n N_i.
\end{eqnarray}

\begin{proposition}
Let $\Ao_1,\ldots,\Ao_n$ be sharp observables with $N_1,\ldots,N_n$ values, respectively. 
There exists a normal multimeter $\langle \K, \Zo, G \rangle$ with $\dim \K=n \cdot \max\{N_1,...,N_n\}$ capable of realizing each of the $\Ao_i$ $i=1,...,n$.
\end{proposition}

\begin{proof}
Denote $d= \max\{N_1,...,N_n\}$. Let $\K = \K_1 \otimes \K_2$ with $\dim \K_1=d$ and $\dim \K_2=n$, and fix some orthonormal bases $\{\phi_k\}_{k=1}^d \subset \K_1$ and $\{\eta_l \}_{l=1}^n  \subset \K_2$. Define the unitaries $G_j: \K_1 \rightarrow \K_1$, $j=1,...,d$, as in Prop. \ref{prop:minimal}.
Extend every observable $\Ao_l$ to be ''$d$-valued'' by adding zero-effects, if necessary, and define the coupling to be 
\begin{equation}
G= \sum_{j=1}^d \sum_{l=1}^n   \Ao_l(j) \otimes G_j^* \otimes P[\eta_l] \, .
\end{equation}
Finally, choose the pointer observable to be $\Zo(k)= |\phi_k\rangle\langle\phi_k| \otimes \id_{\K_2}$. Then 
\begin{eqnarray}
&& \ptr{ G^*\left( \id_\hi \otimes \Zo(k)\right) G \, \id_\hi \otimes P[\phi_1 \otimes \eta_i] } \nonumber \\
&=& \sum_{j,j',l,l'} \Ao_l(j) \Ao_{l'}(j')\, \tr{|G_j \phi_k \rangle\langle G_{j'} \phi_k| \, P[\phi_1]} \,  \tr{ P[\eta_{l}] \, P[\eta_{l'}] \, P[\eta_i]} \nonumber \\
&=& \Ao_i(k).
\end{eqnarray}
This shows, that if one chooses $\Phi_i=\phi_1\otimes \eta_i$, then for every $i=1,...,n$ $\langle \K, \Zo, G, \Phi_i \rangle$ is an $\Ao_i$-measurement.
\end{proof}

Following the methods introduced above one can also study the effectiveness of channel programming. Since one can induce an arbitrary unitary channel $\mc U^* (B) = U^* B U$ from a normal measurement model $\langle \K, G, \phi\rangle$ using a unitary coupling $G = U \otimes \id_\K$, we note that $\min \dim \K =1$.
Interestingly enough, from the point of view of deterministic programming this means that the ''push-a-button'' realization of $n$ unitary channels (see Ex. \ref{ex:push}) is actually the most efficient protocol one can have. This implies that when engineering unitary quantum gate arrays, it's sufficient to build every gate individually and bundle those into an array respecting the ''push-a-button'' protocol.

\section{Conclusions and discussion}\label{sec:discussion}

The fundamental limitations on the deterministic programming of sharp observables or unitary channels are essentially the same.
We have shown how these two scenarios are connected and can be put into a common framework by using the general theory of quantum measurements. In particular, we have generalized the orthogonality result from Nielsen and Chuang, in which the programmable multimeter/gate array is described by a unitary channel, to a completely general case in which the channel is arbitrary. We emphasize, that the original no-go theorem of perfectly precise programmable quantum gate array by Nielsen and Chuang would hold also in this case via dilating the channel into a unitary one. The drawback in this dilation approach is that it doesn't address the orthogonality in the context of original space. Our results clarify this aspect.

This article deals with qualitative aspects of quantum programming, namely the orthogonality of the programming vectors of sharp observables and unitary channels. A quantitative study on the connection between the distance of programming states and the distinction of programmed devices will be the subject of a separate investigation. 
It would also be interesting to answer whether or not a deterministically programmable multimeter capable of measuring every quantum observable is possible, if post-processing is allowed.

\section*{Acknowledgements}

This work has been supported by the Academy of Finland (grant no. 138135). M.T. acknowledges financial support from the University of Turku Graduate School (UTUGS).
The authors are grateful to Jussi Schultz for his comments on an earlier version of this manuscript. M.T. would also like to acknowledge Erkka Haapasalo for many useful discussions.

\end{document}